\documentclass[12pt]{article}%

\usepackage{amssymb}

\usepackage{amsthm}
\usepackage{amsmath}

\usepackage{graphicx,amssymb}
\usepackage{fullpage}
\usepackage{tikz}
\usepackage{xspace}
\usepackage{graphics}
\usepackage{color}
\usepackage[colorlinks,
linkcolor=black,
anchorcolor=black,
citecolor=black,
urlcolor=black
]{hyperref}

\usepackage{multirow}

\usepackage[font=footnotesize]{caption}

\theoremstyle{plain}
\newtheorem{theorem}{Theorem}
\newtheorem{lemma}{Lemma}
\newtheorem{corollary}{Corollary}
\newtheorem{claim}{Claim}

\theoremstyle{definition}
\newtheorem{definition}{Definition}
\theoremstyle{remark}

\newcommand{\xijl}{x_{\langle i, j, j', \l \rangle}}
\newcommand{\yijl}{y_{\langle i, j, j', \l \rangle}}
\newcommand{\zijl}{z_{\langle i, j, j', \l \rangle}}

\renewcommand{\l}{\ell}

\newcommand{\Oh}{\mathcal{O}}

\newcommand{\problem}[1]{\textsc{#1}\xspace}
\newcommand{\fvsp}{\textsc{Feedback Vertex Set}\xspace}
\newcommand{\vcp}{\textsc{Vertex Cover}\xspace}
\newcommand{\hcp}{\textsc{Hamiltonian Cycle}\xspace}
\newcommand{\octp}{\textsc{Odd Cycle Transversal}\xspace}
\newcommand{\cvcp}{\textsc{Connected Vertex Cover}\xspace}
\newcommand{\cdsp}{\textsc{Connected Dominating Set}\xspace}
\newcommand{\stp}{\textsc{Steiner Tree}\xspace}
\newcommand{\dsp}{\textsc{Dominating Set}\xspace}
\newcommand{\imp}{\textsc{Induced Matching}\xspace}
\newcommand{\isp}{\textsc{Independent Set}\xspace}
\newcommand{\cp}{\textsc{Clique}\xspace}
\newcommand{\cnp}{\textsc{Chromatic Number}\xspace}
\newcommand{\pitp}{\textsc{Triangle Partition}\xspace}
\newcommand{\idsp}{\textsc{Independent Dominating Set}\xspace}

\newcommand{\cfa}{\textsc{Conflict-free Assignment}}

\title{Preprocessing Complexity for Some Graph Problems Parameterized by Structural Parameters}

\author{{Manuel Lafond\footnote{Department of Computer Science, Université de Sherbrooke, \textit{Email: manuel.lafond@usherbrooke.ca}}   } \and {Weidong Luo\footnote{Department of Computer Science, Université de Sherbrooke, \textit{Email: weidong.luo@yahoo.com}} }}

\begin{document}

\maketitle

\begin{abstract}
Structural graph parameters play an important role in parameterized complexity, including in kernelization. 
Notably, vertex cover, neighborhood diversity, twin-cover, and modular-width have been studied extensively in the last few years.  However, there are many fundamental problems whose preprocessing complexity is not fully understood under these parameters.  Indeed, the existence of polynomial kernels or polynomial Turing kernels for famous problems such as \cp, \cnp, and \stp has only 
been established for a subset of structural parameters.  In this work, we use several techniques 
to obtain a complete preprocessing complexity landscape for over a dozen of fundamental algorithmic problems.
\end{abstract}


\section{Introduction}
Preprocessing techniques such as kernelization and Turing kernelization form a fundamental branch of parameterized algorithms and complexity \cite{DBLP:books/sp/CyganFKLMPPS15, downey2012parameterized,fomin2019kernelization}.  For most popular algorithmic graph problems, the upper and lower bounds of kernelization are generally well-understood under the natural parameters, i.e. the value to be minimized or maximized.  However, there is much knowledge to be gained for alternate parameters.  
In this work, we aim to complete the kernelization complexity landscape for the structural graph parameters  
vertex cover ($vc$) \cite{DBLP:journals/algorithmica/CyganLPPS14, DBLP:conf/isaac/FellowsLMRS08}, neighborhood diversity ($nd$)  \cite{DBLP:journals/algorithmica/Lampis12, DBLP:journals/lmcs/KnopKMT19}, twin-cover ($tc$) \cite{DBLP:conf/iwpec/Ganian11, DBLP:journals/corr/abs-2302-06983}, and modular-width ($mw$) \cite{DBLP:journals/algorithmica/FominLMT18, DBLP:conf/iwpec/GajarskyLO13}, which all play an important role in parameterized complexity.  
For many famous problems such as \cp, \cnp, and \stp, the existence or non-existence of efficient preprocessing, namely polynomial kernels and polynomial Turing kernels, had only been partially established in previous work (see the entries with a reference in Table~\ref{summary-table-kernels}).  
Using a variety of known and novel techniques, 
we present new preprocessing results for 13 different graph problems, leading to the following theorem.


\begin{table}\label{summary}
\begin{center} 
\begin{tabular}{ | c | c | c | c | c |} 
\hline
Problems $\setminus$ Parameters&  $vc$ & $tc$ &  $nd$ &$mw$\\
\hline

\multirow{2}{*}{\problem{Triangle Partition}}& {PK}  & PK & {PC} & open \\
&   &  &  & no PC \cite{DBLP:journals/dam/Knop20}    \\
\hline

\multirow{2}{*}{\problem{Induced Matching}}&  & & PC & PTC \cite{DBLP:journals/corr/abs-2201-04678} \\
& WK[1]-h  & WK[1]-h  &  &  {no PC}  \\
\hline

\multirow{2}{*}{\problem{Steiner Tree}} &  & & PK \cite{DBLP:journals/corr/abs-2201-04678} & PK \cite{DBLP:journals/corr/abs-2201-04678}\\
&  MK[2]-h & MK[2]-h &  &    \\
\hline

\multirow{2}{*}{\problem{Chromatic Number}}&  & & PC & PTC \cite{DBLP:journals/corr/abs-2201-04678}\\
& {MK[2]-h}  & {MK[2]-h} &  &  {no PC} \\
\hline

\problem{Connected}&  & & PK \cite{DBLP:journals/corr/abs-2201-04678} & PK \cite{DBLP:journals/corr/abs-2201-04678} \\
\problem{Dominating Set}& MK[2]-h & MK[2]-h &  & \\
\hline
			
\multirow{2}{*}{\problem{Hamiltonian Cycle}}& PK \cite{DBLP:journals/tcs/BodlaenderJK13} & PK   \cite{DBLP:journals/tcs/BodlaenderJK13} & PC & PTC \cite{DBLP:journals/corr/abs-2201-04678} \\
&  &   &  &  {no PC}  \\
\hline
	
\multirow{2}{*}{\problem{Clique}}& PTK \cite{DBLP:journals/siamdm/BodlaenderJK14}  & PTK  & PC \cite{DBLP:journals/jcss/GanianSS16} & PTC \cite{DBLP:journals/corr/abs-2201-04678}\\
& no PC \cite{DBLP:journals/siamdm/BodlaenderJK14} & no PC  \cite{DBLP:journals/siamdm/BodlaenderJK14}   & & {no PC}\\
\hline
			
\multirow{2}{*}{\problem{Independent Set}} & PK  \cite{DBLP:journals/jal/ChenKJ01} &  {PK} & PC \cite{DBLP:journals/jcss/GanianSS16} & PTC \cite{DBLP:journals/corr/abs-2201-04678} \\
&  &  &  & {no PC}\\
\hline
			
\multirow{2}{*}{\problem{Vertex Cover}}&  PK \cite{DBLP:journals/jal/ChenKJ01} & PK & PC \cite{DBLP:journals/jcss/GanianSS16} &  PTC \cite{DBLP:journals/corr/abs-2201-04678}\\
&   & &  &  {no PC}\\
\hline

\problem{Feedback} & PK  \cite{DBLP:conf/icalp/Iwata17} & PK \cite{DBLP:journals/dmtcs/Ganian15} & PC \cite{DBLP:journals/jcss/GanianSS16} & PTC \cite{DBLP:journals/corr/abs-2201-04678}  \\
\problem{Vertex Set}&  & &  &  {no PC}     \\
\hline
			 
\problem{Odd Cycle}& PK \cite{DBLP:conf/iwpec/JansenK11} & PK & PC \cite{DBLP:journals/jcss/GanianSS16} & PTC \cite{DBLP:journals/corr/abs-2201-04678} \\
\problem{Transversal}&   &  &  &  {no PC}     \\
\hline

\problem{Connected}&  &  & PC \cite{DBLP:journals/jcss/GanianSS16} & PTC \cite{DBLP:journals/corr/abs-2201-04678} \\
\problem{Vertex Cover}& WK[1]-h \cite{DBLP:journals/algorithmica/HermelinKSWW15}  & WK[1]-h  \cite{DBLP:journals/algorithmica/HermelinKSWW15} &  &  {no PC}     \\
\hline			

\multirow{2}{*}{\problem{Dominating Set}}&  & & PC \cite{DBLP:journals/jcss/GanianSS16}  & PTC \cite{DBLP:journals/corr/abs-2201-04678} \\
& MK[2]-h \cite{DBLP:journals/algorithmica/HermelinKSWW15}  & MK[2]-h  \cite{DBLP:journals/algorithmica/HermelinKSWW15} & &  {no PC} \\
\hline
\end{tabular}
\caption{Preprocessing results for the problems on the first column in parameters vertex cover ($vc$),  twin-cover ($tc$), neighborhood diversity ($nd$), and modular-width ($mw$). The results for each problem include the existence of a polynomial kernel (PK), polynomial compression (PC), polynomial Turing kernel (PTK), and polynomial Turing compression (PTC). Moreover, that a problem is WK[1]-hard (WK[1]-h) or MK[2]-hard (MK[2]-h) means it has a conditional PTC lower bound.  
The first line and second line for each entry contain the positive and negative results of the problem, respectively.
For example, \imp is WK[1]-h in parameter $vc$ and $tc$, has a PC in parameter $nd$, has a PTC but does not admit a PC in parameter $mw$.
In addition, each result with a cited paper means it either comes from the paper or can be obtained straightforwardly using a result of the paper. 
Apart form \textsc{Triangle Partition}$(mw)$, all the problems are known to be FPT \cite{DBLP:journals/dam/Knop20, DBLP:journals/corr/abs-2201-04678}.}
\label{summary-table-kernels}
\end{center}
\end{table}

\begin{theorem}
The results in Table \ref{summary-table-kernels} without references are correct.
\end{theorem}

In particular, although vertex cover is a large graph parameter that usually admits positive preprocessing results, we derive several negative results.  
This is achieved by devising new polynomial parametric transformations, a refined Karp reduction, to prove some problems to be MK/WK-hard.  This means that the problems have no polynomial Turing kernels unless all MK/WK-complete problems in the specific hierarchy have polynomial Turing kernels, and have no polynomial kernels unless coNP $\subseteq$ NP/poly. On the positive side, a polynomial kernel for \pitp is trivial in parameter $vc$, but becomes complicated in parameter $tc$, though feasible as we show.  It appears difficult to create reduction rules for the later parameter, but we are able to reduce it to a kernelizable intermediate problem that we call~\cfa, which is a variant of an assignment problem that may be of independent interest.  So we first produce a polynomial compression  for \textsc{Triangle Partition$(tc)$} by reducing it to the new problem, then reduce the new problem back to \textsc{Triangle Partition} to obtain the polynomial kernel. For the parameter neighborhood diversity, we design a meta-theorem for obtaining the polynomial compressions for all the problems in the table.  As for modular width, we obtain polynomial compression lower bounds using the \textit{cross-composition} technique  \cite{DBLP:journals/siamdm/BodlaenderJK14}, furthermore, we create new original problems when giving the cross-compositions for 
\problem{Connected Vertex Cover($mw$)}, \problem{Odd Cycle Transversal($mw$)}, \problem{Feedback Vertex Set($mw$)}, and \problem{Induced Matching($mw$)}.

\section{Preliminaries}

Let $G = (V, E)$ be a graph.  
For $S \subseteq V$, $G - S$ denotes the subgraph of $G$ induced by $V\setminus S$. For $v\in V$, $N(v)$ denotes the set of neighbors of $v$. 
The parameter \emph{vertex cover} of $G$, denoted by $vc(G)$ or $vc$, is the minimum number of vertices that are incident to all edges of $G$. 
Twin-cover was first proposed in \cite{DBLP:conf/iwpec/Ganian11}. $S\subseteq V$ is a \emph{twin cover} of $G$ if $G - S$ is a cluster, which is formed from the disjoint union of complete graphs, and any two vertices $u,v$ in the same connected component of $G - S$ satisfy  $N(u)\cap S = N(v)\cap S$. The parameter \emph{twin-cover} of $G$, denoted by $tc(G)$ or $tc$, is the cardinality of a minimum twin-cover of $G$. 
Modular-width was first proposed in \cite{DBLP:journals/mst/CourcelleMR00} and was first introduced into parameterized complexity in \cite{DBLP:conf/iwpec/GajarskyLO13}. 
A \textit{module} of $G$ is $M \subseteq V$ such that, for every $v\in V \setminus M$, either $M\subseteq N(v)$ or $M\cap N(v) = \emptyset$. The empty set, $V$, and all $\{v\}$ for $v\in V$ are \textit{trivial modules}. 
A module $M$ is \textit{maximal} if $M\subsetneq V$ and there are no modules $M'$ such that $M \subsetneq M' \subsetneq V$. A module $M$ is \textit{strong} if, for any module $M'$, only one of the following conditions holds: (1) $M\subseteq M'$ (2) $M'\subseteq M$ (3) $M\cap M' = \emptyset$.  
Let $P \subseteq 2^V$ be a vertex partition of $V$. If $P$ only includes modules of $G$, then $P$ is a \textit{modular partition}. 
For two modules $M, M'$ of $ P$, they are \textit{adjacent} if all $v\in M$ are adjacent to all $v'\in M'$, and they are \textit{non-adjacent} if no vertices of $M$ are adjacent to a vertex of $M'$. A modular partition $P$ that only contains maximal strong modules is a \textit{maximal modular partition}. 
For a modular partition $P$ of $V$, we define \textit{quotient graph} $G_{/P} =(V_M, E_P)$ as follows. $v_{M}$ represents the corresponding vertex of a module $M$ of $P$. Vertex set $V_M$ consists of $v_{M}$ for all $M\in P$. For any $M, M'\in P$, edge $v_{M}v_{M'}\in E_P$ if $M$ and $M'$ are adjacent.  
All strong modules $M$ of $G$ can be represented by an inclusion tree $MD(G)$, which is called the \textit{modular decomposition tree} of $G$. Each $M$ is corresponding to a vertex $v_M$ of $MD(G)$. The root vertex $v_V$ of $MD(G)$ corresponds to $V$. Every leaf $v_{\{v\}}$ of $MD(G)$ corresponds a vertex $v$ of $G$. For any two strong modules $M$ and $M'$, $v_{M'}$ is a descendant of $v_M$ in the  inclusion tree iff $M'$ is a proper subset of $M$. 
Consider an internal vertex $v_M$ of $MD(G)$.
If $G[M]$ is disconnected, then $v_M$ is a \textit{parallel} vertex. If $\overline{G[M]}$ is disconnected, then $v_M$ is a \textit{series} vertex. If both $G[M]$ and $\overline{G[M]}$ are connected, then $v_M$ is a \textit{prime} vertex.
The \emph{modular-width} of a graph $G$ is the minimum number $k$ such that the number of children of any prime vertex in $MD(G)$ is at most $k$.
More information about modular-width can be found in \cite{DBLP:conf/iwpec/GajarskyLO13, DBLP:journals/csr/HabibP10}.
Neighborhood diversity was first proposed in \cite{DBLP:journals/algorithmica/Lampis12}.
A modular partition of $V$ is called a \textit{neighborhood partition} if every module in the modular partition is either a clique or an independent set, which are called \textit{clique type} and  \textit{independent type}, respectively. The width of the partition is its cardinality. The \textit{minimum neighborhood partition} of $V$, which can be obtained in polynomial time \cite{DBLP:journals/algorithmica/Lampis12}, is the \textit{neighborhood partition} of $V$ with the minimum width $k$.  
The \textit{neighborhood diversity}, denoted by $nd(G)$ or $nd$, of $G$ is the width of the minimum neighborhood partition of $V$. 
Based on Theorem 3 of \cite{DBLP:conf/iwpec/GajarskyLO13} and Theorem 7 of \cite{DBLP:journals/algorithmica/Lampis12}, $mw(G)\leq nd(G)\leq O(2^{vc(G)})$. Based on Definition 3.1 of \cite{DBLP:conf/iwpec/Ganian11} and Theorem 3 of \cite{DBLP:conf/iwpec/GajarskyLO13}, $mw(G)\leq O(2^{tc(G)})$ and $tc(G) \leq vc(G)$, which means that the negative results in this paper in parameter vertex cover also hold in parameter twin-cover. 
These parameters are related as in Figure \ref{fig:mw-nd-tc-ve-relationship}.

A parameterized problem is a language $Q\subseteq \Sigma^* \times \mathbb{N}$. We call $k$ the parameter if $(x,k) \in Q$. $Q$ is fixed-parameter tractable (FPT) if it is decidable in $f(k)|x|^{O(1)}$ time for some computable function $f$. $Q(vc)$ denotes $Q$ parameterized by vertex cover of its input graph. The definitions of $Q(tc)$, $Q(nd)$, and $Q(mw)$ go the same way.

\begin{definition}
A \emph{kernelization} (\emph{compression})  for a parameterized problem $Q$ is a polynomial-time algorithm, which takes an instance $(x,k)$ and produces an instance $(x',k')$ called a kernel, such that $(x,k) \in Q$ iff $(x',k') \in Q$ ($(x',k') \in Q'$ for a problem $Q'$), and the size of $(x',k')$ is bounded by a computable function $f$ in $k$. Moreover, we say $Q$ admits a kernel or kernelization (compression) of  $f(k)$ size.
\end{definition}

In particular, $Q$ admits a polynomial kernel or polynomial kernelization (PK) (polynomial compression (PC)) if $f$ is polynomial. Compared to kernelization, compression allows the output instance to belong to any problem.

\begin{definition}
A \emph{Turing kernelization} (\emph{Turing compression}) for a parameterized problem $Q$ is a polynomial-time algorithm with the ability to access an oracle for $Q$ (a problem $Q'$) that can decide whether $(x,k)\in Q$ with queries of size at most a computable function $f$ in $k$, where the queries are called Turing kernels. Moreover, we say $Q$ admits a Turing kernel or Turing kernelization (Turing compression) of $f(k)$ size.
\end{definition}

\begin{figure}
\centering
\begin{tikzpicture}
\draw[color=black] (-2,1.2) node[scale=0.9] {$mw$};
\draw[-stealth]  (-3,0.2) -- (-2.05,1);
\draw[-stealth] [dashed] (-1,0.2) -- (-1.95,1);

\draw[color=black] (-3,0) node[scale=0.9] {$nd$};
\draw[color=black] (-1,0) node[scale=0.9] {$tc$};

\draw[stealth-]  (-1,-0.2) -- (-1.95,-1);
\draw[stealth-] [dashed] (-3,-0.2) -- (-2.05,-1);

\draw[color=black] (-2,-1.2) node[scale=0.9] {$vc$};

\draw[color=black] (2,1.2) node[scale=0.9] {PTC};
\draw[-stealth] (3,0.2) -- (2.05,1);
\draw[-stealth] (1,0.2) -- (1.95,1);

\draw[color=black] (3,0) node[scale=0.9] {PTK};
\draw[color=black] (1,0) node[scale=0.9] {PC};

\draw[stealth-] (1,-0.2) -- (1.95,-1);
\draw[stealth-] (3,-0.2) -- (2.05,-1);

\draw[color=black] (2,-1.2) node[scale=0.9] {PK};

\end{tikzpicture}
\caption{\footnotesize{The left and the right figures are the relations among the parameters and the types of preprocessing, respectively, discussed in this paper. The left figure includes parameters modular-width ($mw$), neighborhood diversity ($nd$), twin-cover ($tc$), and vertex cover ($vc$). The right figure includes polynomial Turing compression (PTC), polynomial Turing kernelization (PTK), polynomial compression (PC), and polynomial kernelization (PK). The arrows indicate generalization, e.g. PTC generalizes PTK and thus a problem has a PTK implies it has a PTC; $mw$ generalizes $tc$ and thus is bounded by a function $f$ in $tc$, more specifically, solid and dashed arrows imply linear function $f$ and exponential function $f$, respectively.}}
\label{fig:mw-nd-tc-ve-relationship}
\end{figure}
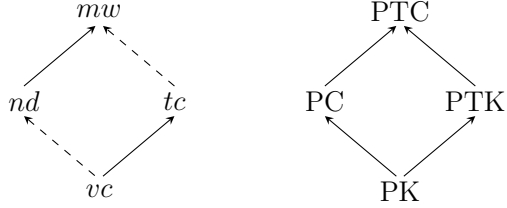

In particular, $Q$ admits a polynomial Turing kernel or a polynomial Turing kernelization (PTK) (polynomial Turing compression (PTC)) if $f$ is a polynomial function. The relations among these preprocessing operations can be found in Figure \ref{fig:mw-nd-tc-ve-relationship}. Finding a framework to rule out PTC or even PTK under some widely believed complexity hypothesis is a long-standing open problem \cite{fomin2019kernelization}. PTCs of several special FPT problems are refuted under the assumption that exponential hierarchy does not collapse \cite{DBLP:journals/tcs/Luo22}. 
A completeness theory for PTK (PTC) was proposed in \cite{DBLP:journals/algorithmica/HermelinKSWW15}, which constructs a WK/MK-hierarchy and demonstrates a large number of problems in this hierarchy. 
Moreover, all the problems in WK[i]-complete or MK[j]-complete have PTKs (PTCs) if any problem in WK[i]-hard or MK[j]-hard has a PTK (PTC), where $i\geq 1$ and $j\geq 2$. 
For a parameterized problem $Q$, we can provide a polynomial parametric transformation (PPT) from a known WK/MK-hard problem to $Q$ to demonstrate $Q$ is WK/MK-hard, which gives evidence that $Q$ does not have a PTK (PTC). The definition of PPT is as follows.

\begin{definition}[\cite{DBLP:journals/tcs/BodlaenderTY11, DBLP:journals/algorithmica/HermelinKSWW15}]
A \emph{polynomial parametric transformation} (PPT) is a  polynomial-time many-one reduction between two parameterized problems such that the parameter of the output instance is polynomially bounded by the parameter of the input instance.
\end{definition}

Next, we provide the definitions of some problems discussed in this paper, all of which are NP-hard. Given a graph $G$, \textsc{Hamiltonian Cycle} (\textsc{Hamiltonian Path}) asks whether $G$ contains a cycle (path) that visits each vertex of $G$ exactly once. Given a formula $\Phi$ of the conjunctive normal form (CNF), \textsc{CNF-SAT} asks whether $\Phi$ has a satisfying assignment.
Given a graph $G$, an integer $k$, and a coloring function $c : V(G) \rightarrow [k]$. \textsc{Multicolored Clique} asks whether $G$ has a multicolored clique of size $k$, i.e. a clique with $k$ vertices colored by $k$ different colors.
Assume $(G,k)$ is the input of the following problems, where $G=(V,E)$ is a graph and $k$ is an integer. 
\textsc{Chromatic Number} asks whether there exists at most $k$ colors to color the vertices of $G$ such that no two adjacent vertices share the same color.
\textsc{Clique} asks whether there exists a set $X$ of at least $k$ vertices of $G$ such that any two vertices in $X$ are adjacent.
\textsc{Connected Dominating Set} asks whether there exists a set $X$ of at most $k$ vertices of $G$ such that $G[X]$ is connected and every vertex not in $X$ is adjacent to at least one member of $X$.
\textsc{Connected Vertex Cover} asks whether there exists a set $X$ of at most $k$ vertices of $G$ such that $G[X]$ is connected and every edge of $G$ has at least one endpoint in $X$.
\textsc{Dominating Set} asks whether there exists a set $X$ of at most $k$ vertices of $G$ such that every vertex not in $X$ is adjacent to at least one member of $X$.
\textsc{Feedback Vertex Set} asks whether there exists a set $X$ of at most $k$ vertices of $G$ such that $G - X$ is a forest.
\textsc{Independent Set} asks whether there exists a set $X$ of at least $k$ vertices of $G$ such that any two vertices in $X$ have no edge.
\textsc{Induced Matching} asks whether there exists a set $X$ of $2k$ vertices of $G$ such that the subgraph induced by $X$ is a matching consisting of $k$ edges.
\textsc{Odd Cycle Transversal} asks whether there exists a set $X$ of at most $k$ vertices of $G$ such that $G - X$ is a bipartite graph.
\textsc{Triangle Partition} asks whether $G$ contains $k/3$ vertex disjoint triangles whose union includes every vertex of $G$.
\textsc{Steiner Tree} asks whether there exists a connected subgraph of $G$ that contains at most $k$ edges and contains all vertices of $K$.
\textsc{Vertex Cover} asks whether there exists a set $X$ of at most $k$ vertices of $G$ such that every edge of $G$ has at least one endpoint in $X$.

\section{Parameterization by vertex cover number}
\label{par-vertex-cover-section}
\textsc{Hamiltonian Cycle$(vc)$} \cite{DBLP:journals/tcs/BodlaenderJK13}, \textsc{Odd Cycle Transversal$(vc)$}
\cite{DBLP:conf/iwpec/JansenK11} have PKs.
A PK for \textsc{Vertex Cover$(vc)$} \cite{DBLP:journals/jal/ChenKJ01} implies a PK for \textsc{Independent Set$(vc)$} since, for a graph, the vertex cover number pluses the independence number equals the vertex number. Clearly, 
A PK for \textsc{Feedback Vertex Set($k$)}  \cite{DBLP:conf/icalp/Iwata17} implies a PK for \textsc{Feedback Vertex Set$(vc)$} since the vertex cover number is at least the feedback vertex set number. 
Furthermore, the WK[1]-hardness of \textsc{Connected Vertex Cover($k$)}  \cite{DBLP:journals/algorithmica/HermelinKSWW15} implies the WK[1]-hardness of \textsc{Connected Vertex Cover($vc$)} since the vertex cover number is at most the connected vertex cover number. 
A PPT from a WK[1]-hard problem \textsc{Multicolored-Clique}($k \log n$) \cite{DBLP:journals/algorithmica/HermelinKSWW15} to \textsc{Induced Matching}$(vc)$ is provided as follows.

\begin{theorem}
\textsc{Induced Matching}$(vc)$ is WK[1]-hard.
\end{theorem}
\begin{proof}
Let $(G, c)$ be an instance of \textsc{Multicolored-Clique}($k \log n$), where $G = (V, E)$ and $c : V \rightarrow [k]$.  Assume that $V = [n]$ and, for $v \in V$, denote by $b_{\ell}(v) \in \{0, 1\}$ the $\l$-th bit of the binary representation of $v$. Without loss of generality, suppose $\log n$ is an integer.
We construct an instance $H$ of \textsc{Induced Matching}($vc$). The vertex set is $V(H) = E \cup F \cup Z$, and the graph can be described as follows:
(1) $E$ is the set of edges of $G$.  That is, each edge of $G$ has a corresponding vertex in $H$, and they form an independent set in $H$.
(2) For $i, j \in [k]$ with $i<j$, add a vertex $f_{i,j}$ to $F$. The neighbors of $f_{i,j}$ are $N(f_{i,j}) = \{ uv \in E : c(u) = i, c(v) = j \}$.
(3)  For every pair of pairs $(i, j), (i, j') \in [k] \times [k]$, with $i, j, j'$ distinct, and for every $\ell \in [\log n]$, add the vertices $\xijl, \yijl, \zijl$ to $Z$.
Their neighborhoods are $N(\xijl) = \{uv \in E : c(u) = i, c(v) = j, b_\l(u) = 1\} \cup \{\yijl,\zijl \},
N(\yijl) = \{uv \in E : c(u) = i, c(v) = j', b_\l(u) = 0\} \cup \{\xijl, \zijl\},$  and $N(\zijl) = \{\xijl, \yijl\}.$
Notice that $\xijl, \yijl, \zijl$ induce a triangle in $Z$.

We claim that $G$ contains a multicolored clique iff $H$ contains an induced matching of size $\binom{k}{2}  + |Z|/3$.  Note that $F \cup Z$ is a vertex cover of $H$. 
 Since $|F| = \binom{k}{2}$ and $|Z| \in O(k^3 \log n)$, the size of the parameter is polynomial in $k \log n$.
Let $v_1, \ldots, v_k$ be vertices of a multicolored clique of $G$, where $c(v_i) = i$ for every $i\in [k]$.  We construct an induced matching $I$ of $H$ in two steps.  In step 1, for each edge $v_iv_j$ in the clique, $i < j$, add to $I$ the edge $\{ v_i v_j, f_{i,j} \}$ (one can easily check that these edges are independent).  In step 2, consider values $i, j, j', \l$ and the corresponding triangle in $Z$.  If $b_\l(v_i) = 0$, then there is no edge in $H$ between any $v_i v_{j''}$ and $\xijl \in Z$, for $j'' \neq i$.  There is also no edge in $H$ between any $v_jv_{j'}$ and $\xijl$, for any $j, j' \neq i$. Hence we may add $\xijl \zijl$ to $I$, as it is independent from the edges added in step 1.
Similarly, if $b_\l(v_i) = 1$, we may add $\yijl \zijl$ to $I$ instead.
Notice that any two edges that we add in step 2 are independent since no edges are shared between any two triangles in $Z$.
It follows that $I$ is an induced matching, and its size is exactly $\binom{k}{2} + |Z|/3$.

Suppose that $H$ contains an induced matching $I$ of size $\binom{k}{2} + |Z|/3$. 
Notice that for $i, j, j', \l$, at most one edge incident to an element of $\{\xijl, \yijl, \zijl\}$ can be in $I$, since they form a triangle.
Therefore these can account for at most $|Z|/3$ edges of $I$.  Since $E$ is an independent set in $H$, the remaining $\binom{k}{2}$ edges of $I$ can only be edges incident to the $f_{i,j}$ vertices of $F$.
We may thus assume that each $f_{i,j}$ is incident to an edge of $I$.  In turn, we may assume that for each triangle in $Z$, at least one of its vertices is incident to an edge in $I$.
We claim that $C = \{ uv : i, j \in [k], \{uv, f_{i,j}\} \in I \}$ form the edges of a multicolored clique of $G$.  
Notice that for each $i \neq j$, there exists exactly one edge $uv \in C$ such that $c(u) = i, c(v) = j$.

Now let $i \in [k]$ and assume, towards a contradiction, that there are $uv, u'v' \in C$ such that $c(u) = c(u') = i$ but $u \neq u'$.  Let $j = c(v), j' = c(v')$ (by the previous remark we may assume $j \neq j'$).
Since $u \neq u'$, there is some $\l$ such that $b_\l(u) \neq b_\l(u')$, say $b_\l(u) = 1, b_\l(u') = 0$, without loss of generality.
Then in $H$, $\xijl$ is a neighbor of $uv$ and $\yijl$ is a neighbor of $u'v'$. 
Since $\{uv, f_{ij}\}, \{u'v', f_{i,j'}\} \in I$, neither $\xijl$ or $\yijl$ can be incident to an edge in $I$ (see Figure~\ref{fig:induced-matching}). 
This implies that $\zijl$ is also not incident to such an edge.  That is, no vertex of the triangle for values $i, j, j', \l$ is incident to an edge of $I$, a contradiction.  
\end{proof}

\begin{figure}
\centering
\begin{tikzpicture}

\filldraw[color=blue, fill=blue!3] (0,0) ellipse (0.6 and 1.1); 
\filldraw[color=blue, fill=blue!3] (3,0) ellipse (0.8 and 1.2); 

\filldraw[color=blue, fill=blue!3] (7,0) ellipse (1.8 and 1.2); 

\draw[color=black] (0,1.5) node {$F$};
\draw[color=black] (7,1.5) node {$Z$};
\draw[color=black] (3,1.5) node {$E$};

\draw[color=black] (0,0.8) node {$f_{i,j}$};
\fill [color=black] (0,0.4) circle (1.5pt);
\draw[color=black] (0,-0.3) node {$f_{i,j'}$};
\fill [color=black] (0,-0.6) circle (1.5pt);

\draw (0,0.4) -- (3,1);
\draw (0,-0.6) -- (3,-0.2);
\draw (0,-0.6) -- (3,-1);

\draw[color=black] (3,0.75) node {$uv$};
\fill [color=black] (3,1) circle (1.5pt);
\draw[color=black] (3,0.1) node {$uv''$};
\fill [color=black] (3,-0.2) circle (1.5pt);
\draw[color=black] (3,-0.75) node {$u'v'$};
\fill [color=black] (3,-1) circle (1.5pt);

\draw (6.2,0.4) -- (3,1);
\draw (6.2,-0.4) -- (3,-1);

\draw[color=black] (6.3,0.7) node {$\xijl$};
\draw[color=black] (6.3,-0.7) node {$\yijl$};
\draw[color=black] (7.9,0) node {$\zijl$};

\fill [color=black] (6.2,0.4) circle (1.5pt);
\fill [color=black] (6.2,-0.4) circle (1.5pt);
\fill [color=black] (7,0) circle (1.5pt);
\draw (6.2,0.4) -- (6.2,-0.4);
\draw (6.2,-0.4) -- (7,0);
\draw (7,0) -- (6.2,0.4);

\end{tikzpicture}
\caption{Illustration of the reduction of \imp problem. The color of $u$ and $u'$ is $i$. The color of $v$, $v'$, and $v''$ are $j$, $j'$, and $j'$, respectively. In addition, $b_l(u) = 1$ and $b_l(u') = 0$.}
\label{fig:induced-matching}
\end{figure}
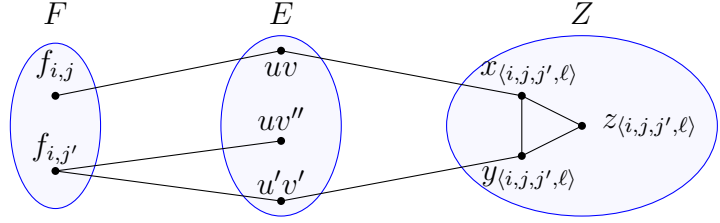

\begin{corollary}
\textsc{Induced Matching}$(vc)$ does not admit a PC unless coNP $\subseteq$ NP/poly.
\end{corollary}

We will provide PPTs from the MK[2]-hard problem CNF-SAT($n$) \cite{DBLP:journals/algorithmica/HermelinKSWW15} to the problems in Theorem \ref{cn-hard-kerenl}, \ref{ST-is-MK-hard}, and \ref{cds-hard}. Suppose w.o.l.g that the formula $\Phi$ of CNF-SAT($n$) has no duplicate clauses and each clause has no duplicate variables. Let $x_1,\ldots,x_n$ and $C_1,\ldots,C_m$ be the variables and the clauses of $\Phi$, respectively.
It has been proved \cite{DBLP:journals/siamdm/BodlaenderJK14} that \textsc{Clique}($vc$) and \textsc{Chromatic Number}($vc$) do not admit PCs unless coNP $\subseteq$ NP/poly, but, clique($vc$) has a PTK. For the PTK of \textsc{Chromatic Number($vc$)}, we give the following negative result.

\begin{theorem}
\label{cn-hard-kerenl}
\textsc{Chromatic Number}$(vc)$ is MK[2]-hard.
\end{theorem}

\begin{proof}
Given an instance $\Phi$ of CNF-SAT, construct a graph $G$ as follows. In what follows, we assume that $i,k \in [n]$ and $j \in [m]$ (e.g. ``for all $i$'' means for all $i \in [n]$). First, add a complete graph with $n$ vertices $v_1,\ldots,v_n$. Secondly, add vertices $w_i,\overline{w}_i$ and edge $w_i\overline{w}_i$ for all $i$, and edges $w_iv_k$, $\overline{w}_iv_k$ for all $i$ and $k$ if $i\neq k$. Thirdly, add vertices $t_1,\ldots,t_m$, moreover, for all $i,j$, add $t_jw_i$ if literal $x_i$ is not in the clause $C_j$ of the instance $\Phi$, and add $t_j\overline{w}_i$  if literal $\overline{x}_i$ is not in the clause $C_j$ of the instance $\Phi$. Finally, add a vertex $u$ and edges $u v_i$, $ut_j$ for all $i$ and $j$.
The purpose of $u$ is to limit the number of possible colors that the $t_i$ vertices can use.
Clearly, all vertices of sets $\{v_1,\ldots,v_n\}$, $\{w_1,\ldots,w_n\}$, and $\{\overline{w}_1,\ldots,\overline{w}_n\}$ together with vertex $u$ are a vertex cover of $G$, so $vc(G) \leq 3n+1$.

We claim that $\Phi$ is satisfiable iff $G$ can be colored with $n+1$ colors.
Assume $\Phi$ is satisfiable. There is an assignment for variables $x_1,\ldots,x_n$ such that $\Phi$ is evaluated to true. 
First, each $v_i$ is colored by $c_i$, and $u$ is colored by $c$. For every pair $w_i$ and $\overline{w}_i$, assign $c_i$ to the vertex whose corresponding literal is true in the assignment of $\Phi$, and the color $c$ to the other vertex. Every $C_j$ of $\Phi$ has at least one literal assigned to true, say $x_i$ (or $\overline{x}_i$ ), thus, some neighbor $\overline{w}_i$ (or $w_i$) of $t_j$ is colored by $c$, and we can assign color $c_i$ to $t_j$. As a result, graph $G$ is colored by $n+1$ colors.
For the reverse direction, assume $G$ is colored by colors $c_1,\ldots,c_n$, and $c$. The subgraph induced by $v_1,\ldots,v_n$ and $u$ is a complete graph with $n+1$ vertices, so $n+1$ colors are needed to color these vertices. Suppose w.l.o.g that $v_i$ is colored by $c_i$ for every $i$, and $u$ is colored by $c$. For every pair $w_i$ and $\overline{w}_i$, their colors are different and come from $\{c_i,c\}$.
Each $t_j$ is adjacent to all $w_i$ (or $\overline{w}_i$) if the literal $x_i$ (or $\overline{x}_i$) is not in $C_j$. Thus, in every $C_j$ of $\Phi$, there is at least one literal whose corresponding vertex is not colored by $c$, otherwise, we need at least $n+1$ colors to color the neighbor vertices of $t_j$, and $n+2$ colors to color $G$. As a result, for every pair $w_i$ and $\overline{w}_i$, choose the vertex that is not colored by $c$ and assign its corresponding literal true, then all clauses of $\Phi$ are satisfied.
\end{proof}

Note that Theorem \ref{cn-hard-kerenl} also implies \textsc{Chromatic Number}($vc$) has no PC unless NP $\subseteq$ coNP/poly, which has been proved in \cite{DBLP:journals/siamdm/BodlaenderJK14}, while our proof is much shorter.

\begin{theorem}\label{ST-is-MK-hard}
\textsc{Steiner Tree}$(vc)$ is MK[2]-hard.
\label{st}
\end{theorem}

\begin{proof}
Let $i \in [n]$ and $j\in [m]$. Given an instance $\Phi$ of CNF-SAT. Construct a graph $G$ as follows. First, add vertices $c_1,\ldots,c_m,$ vertex $ u$, and $n$ triangles to $G$, where the vertices of the $i$-th triangle are $u_i$, $\overline{u}_i$ and $v_i$ for all $i$. Then, for all $i$ and $j$, add edge $u_ic_j$ if literal $x_i$ is in $C_j$ of $\Phi$, and add edge $\overline{u}_ic_j$   if literal $\overline{x}_i$ is in $C_j$ of $\Phi$. Moreover, connect $u$ with all $u_i$ and all $\overline{u}_i$. In addition, the terminal set $K$ consists of all $v_i$, all $c_j$, and $u$. Clearly, the vertices $\{u_1,\ldots,u_n\}  \cup \{\overline{u}_1,\ldots,\overline{u}_n\}$ form a vertex cover of $G$, thus, $vc(G) \leq 2n$.
	
We claim that $\Phi$ is satisfiable iff there is a connected subgraph $T$ of $G$ that contains at most $2n+m$ edges and all vertices of $K$.
Assume $\Phi$ is satisfiable. There is an assignment for variables $x_1,\ldots,x_n$ that $\Phi$ is evaluated true. Consider the vertex set $V(T)$ of $T$. First, add all vertices of $K$ to $V(T)$. Secondly, for all $i$, add vertex $u_i$ to $V(T)$ if the corresponding literal $x_i$ is evaluated true in the assignment, and add vertex $\overline{u}_i$ to $V(T)$ if the corresponding literal $\overline{x}_i$ is evaluated true in the assignment. Since $\Phi$ is satisfiable, at least one vertex of $V(T)\setminus K$ is adjacent to $c_j$ for every $j$, which means that $N(c_j) \cap (V(T)\setminus K) \neq \emptyset$.
 Consider the edge set $E(T)$ of $T$. First, add edge $uv$ to $E(T)$ for all $v\in V(T)\setminus K$. Secondly, for every $c_j$, add the edge between $c_j$ and any one vertex of $N(c_j) \cap (V(T)\setminus K)$ to $E(T)$. Thirdly, since exact one of $x_i$ and $\overline{x}_i$ is evaluated true, exact one of $u_i$ and $\overline{u}_i$ is in $V(T)$ for every $i$. We may add an edge between $v_i$ and the vertex in $N(v_i) \cap V(T)$ to $E(T)$ for every $i$.  Clearly, $T$ is a tree with $2n+m$ edges and $K \subseteq V(T)$.
For the reverse direction, if there is a connected subgraph $T$ of $G$ that contains at most $2n+m$ edges and contains all vertices of $K$. It means that $T$ contains at most $2n+m+1$ vertices, otherwise, $T$ is not connected. 
 Suppose vertex set $U$ consists of all $u_i$ and $\overline{u}_i$. Clearly, we have $V(G) = K \cup U$. 
Consider every pair of $u_i$ and $\overline{u}_i$. Since $v_i$ is an isolated vertex in $G - \{u_i, \overline{u}_i\}$, at least one vertex of $u_i$ and $\overline{u}_i$ is in $V(T)$. In addition, $K$ has $n+m+1$ vertices and is a subset of $V(T)$, so $V(T) \cap (V(G)\setminus K) = V(T) \cap U$ contains at most $n$ vertices. Thus, exact one of $u_i$ and $\overline{u}_i$ is in $V(T)$ for every $i$. Since $T$ is connected and $N(c_j) \subseteq U$ for every $j$, $N(c_j) \cap (V(T)\cap U) \neq \emptyset$ for every $j$.
Now, we may give an assignment to $\Phi$ as follows. For every $i$, assign true to variable $x_i$ of $\Phi$ if vertex $u_i$ is in $V(T)\cap U$, and assign false to variable $x_i$ of $\Phi$ if vertex $\overline{u}_i$ is in $V(T)\cap U$. Since every $c_i$ has at least one adjacent vertex in $V(T)\cap U$, every clause of $\Phi$ is evaluated true under that assignment.
\end{proof}

\begin{corollary}
\textsc{Steiner Tree$(vc)$} does not admit a PC unless NP $\subseteq$ coNP/poly.
\end{corollary}

It is known that dominating set($vc$) is MK[2]-hard \cite{DBLP:conf/icalp/DomLS09, DBLP:journals/algorithmica/HermelinKSWW15}. In the following, we study the \textsc{Connected Dominating Set}$(vc)$.

\begin{theorem}
\label{cds-hard}
 \textsc{Connected Dominating Set}$(vc)$ is MK[2]-hard.
\end{theorem}
\begin{proof}
Let $f_{cds}(G)$ represent the size of the minimum connected dominating set (min-CDS) of $G$. The construction here is similar to that of Theorem \ref{st}. Assume $i\in [n], j \in [m]$.  Given an instance $\Phi$ of CNF-SAT. Construct a graph $G$ as follows. First, add vertices $c_1,\ldots,c_m$ as well as $n$ triangles, and the vertices of the $i$-th triangle are $u_i$, $\overline{u}_i$, and $v_i$ for all $i$. Secondly, for all $i, j$, add edge $u_ic_j$ if literal $x_i$ is in clause $C_j$ of $\Phi$, and add edge $\overline{u}_ic_j$ if literal $\overline{x}_i$ is in clause $C_j$ of $\Phi$. 
Thirdly, add a vertex $p$ and edges $pu_i$, $p\overline{u}_i$ for all $i$. Finally, add a vertex $q$ and the edge $pq$. Clearly, all the vertices of $\{u_1,\ldots,u_n\}$ and $\{\overline{u}_1,\ldots,\overline{u}_n\}$ as well as vertex $p$ are a vertex cover of $G$. Thus, $vc(G) \leq 2n+1$.

We claim $\Phi$ is satisfiable iff $f_{cds}(G) \leq n+1$.
Assume $\Phi$ is satisfiable. There is an assignment for $x_1,\ldots,x_n$ that $\Phi$ is evaluated true. First, add $p$ to a vertex set $D$. Then, for all $i$, add vertex $u_i$ to $D$ if the corresponding literal $x_i$ is evaluated true in the assignment, and add vertex $\overline{u}_i$ to $D$ if the corresponding literal $\overline{x}_i$ is evaluated true in the assignment. Since all $u_i$ and $\overline{u}_i$ are neighbors of $p$, $D$ is a connected dominating set of $G$, moreover, the size of $D$ is $n+1$.
For the other direction, assume $f_{cds}(G) \leq n+1$. Suppose $D$ is a min-CDS of $G$. If $p,q \in D$, then $D \setminus \{q\}$ is a min-CDS, a contradiction. If $p,q \notin D$, then $q$ is not dominated by $D$ and is not in $D$, a contradiction. If $p\notin D$ and $q \in D$, then the subgraph induced in $G$ by $D$ is not connected, a contradiction. As a result, $q\notin D$ and $p \in D$. For every triangle $u_i-\overline{u}_i-v_i$, there is at least one vertex in $D$, otherwise, $v_i$ is not dominated by $D$ or is not in $D$. Thus, the size of $D$ is $n+1$, moreover, there is exactly one vertex from each triangle. Consider each triangle. If $v_i\in D$, then $u_i, \overline{u}_i \not\in D$, and $G[D]$ is disconnected. Thus, $v_i \not \in D$, and $D$ includes exactly one of $u_i, \overline{u}_i$.
For every $i$, assign true to $x_i$ of $\Phi$ if $u_i \not \in D$, and assign false to $x_i$ of $\Phi$ if $\overline{u}_i \in D$. Since $D$ is a dominating set of $G$, any clause of $\Phi$ has at least one literal assigned true.
\end{proof}

\begin{corollary}
 \textsc{Connected Dominating Set}$(vc)$ does not admit a PC unless NP $\subseteq$ coNP/poly.
\end{corollary}

Note that \dsp is MK[2]-hard even parameterized by vertex cover number and dominating set number \cite{DBLP:conf/icalp/DomLS09, DBLP:journals/algorithmica/HermelinKSWW15}. This also holds for our results, because the connected dominating set number is less than the vertex cover number of $G$ in our proof. It is worth mentioning that \idsp is also MK[2]-hard using a very similar proof as that of Theorem \ref{cds-hard}, where the only difference is that vertices $p,q$ are not needed in the new construction.

In addition, a linear kernel for \textsc{Triangle Partition}$(vc)$ is trivial since at most $\frac{vc}{2}$ vertices can be outside of the vertex cover. However, a PK for \textsc{Triangle Partition}$(tc)$ becomes complicated, which is shown in Theorem \ref{triangle-partition-pk-in-tc}.

\section{Parameterization by twin-cover}

A PK of \textsc{Hamiltonian Cycle} parameterized by the \textit{vertex-deletion distance from a cluster}, which is at most $tc$, is given in \cite{DBLP:journals/tcs/BodlaenderJK13}. 
Thus, \textsc{Hamiltonian Cycle$(tc)$} has a PK. \textsc{Feedback Vertex Set$(tc)$} admits a PK  \cite{DBLP:journals/dmtcs/Ganian15}.
In addition, we know $tc(G) \leq vc(G)$. Therefore, the hardness results for the problems parameterized by $vc$ in section \ref{par-vertex-cover-section} also hold for those problems parameterized by $tc$.
We now turn to \textsc{Triangle Partition($tc$)}.  As we already mentioned, reduction rules seem difficult to produce, but we can reduce the problem to a kernelizable intermediate problem that we call~\cfa. We first give its definition and PK as follows.

\medskip
\noindent
\textbf{Input}: a  bipartite graph $G = (B \cup S, E)$, where $B$ and $S$ are respectively called buyers and sellers; a profit  function $p : B \rightarrow \mathbb{N}$; a weight function $w : B \rightarrow  \mathbb{N}$; a capacity function $c : S \rightarrow \mathbb{N}$; a set $P \subseteq \binom{B}{2}$ called conflicting pairs; an integer $q$. In addition, the weights, profits, and capacities are at most exponential in $|B|$, that is $2^{poly(|B|)}$.

\medskip
\noindent
\textbf{Question}: does there exist a subset $F \subseteq E$ of edges satisfying that  
(1) each $b \in B$ is incident to at most one edge of $F$ (each buyer is assigned to at most one seller);
(2) for each $s \in S$, $\sum_{b : bs \in F} w(b) \leq c(s)$ (each seller is assigned buyers of at most its capacity);
(3) for each $\{b_1, b_2\} \in P$, at most one of $b_1$ or $b_2$ is incident to an edge of $F$ (no pair of conflicting buyers is assigned);
(4) $\sum_{b \in V(F) \cap B} p(b) \geq q$ (profit of assigned buyers is at least $q$, where here $V(F)$ is the set of vertices incident to an edge of $F$).

Intuitively, each buyer $b$ incurs profit $p(b)$ if we can assign it to a seller.  However, $b$ has a weight of $w(b)$, and each seller $s$ can accommodate a total weight of $c(s)$.  
Moreover, $P$ specifies pairs of buyers that cannot both be assigned (even forbidding assigning them to different sellers).  We want to assign buyers under these constraints to achieve a profit of at least $q$. Clearly, \cfa~is in NP.

\begin{lemma}
\label{kernel-cfa-para-B-full-app}
\cfa~has PK parameterized by $|B|$.
\end{lemma}
\begin{proof}
The reduction rules are as follows. (1) Delete vertices of degree $0$. (2) Delete edges $bs \in E$ such that $w(b) > c(s)$. (3) Assume that rule 2 is not applicable and that there is $b \in B$ of degree at least $|B| + 1$.  Then remove every edge of $E$ incident to $b$, add a new vertex $s$ in $S$ of capacity $c(s) = w(b)$, and add the edge $bs$ to $E$. Clearly, rules 1 and 2 are safe. Rule 3 is safe because if we have a solution for the original instance, then we may use it directly for the modified instance if $b$ is not assigned.  If $b$ is assigned to some seller, in the modified instance we just assign $b$ to the newly created $s$, achieving the same profit. 
Conversely, if we have a solution to the modified instance, we use the same solution for the original instance --- unless $b$ is assigned to $s$.  In that case, since $b$ has more than $|B|$ neighbors in the original instance, one of them has no assigned buyer in the solution to the modified instance, and we may reassign $b$ to this free neighbor.

Suppose we have applied the above rules until exhaustion. We know that every vertex of $B$ has degree at most $|B|$ by Rule 2 and Rule 3 and, since there are no isolated vertices by Rule 1, it follows that $|S| \leq |B|^2$.
Thus, the number of vertices is at most $|B| + |B|^2$.
Since the weights, profits, and capacities are exponential in $|B|$, each of them can be encoded in $|B|^{O(1)}$ bits. Thus, the size of the kernel is a polynomial in $|B|$.
\end{proof}

\begin{theorem}
\label{triangle-partition-pk-in-tc}
\textsc{Triangle Partition($tc$)} has a PK.
\end{theorem}
\begin{proof}
We first provide a Karp reduction from \textsc{Triangle Partition($tc$)} to \cfa. Given an instance $(G, k)$ of \textsc{Triangle Partition($tc$)} together with a minimum twin-cover $X$ of $G=(V,E)$, where $|X| = tc$, let $C = \{C_1, \ldots, C_m\}$ be the partition of $V \setminus X$ into cliques.
Let $i\in [m]$. For each $C_i \in C$, if $C_i$ has more than $2tc+2$ vertices, then we repeatedly delete any $3$ vertices from it until its size satisfies that $2tc \leq |C_i| \leq  2tc+2$. The reduction is safe because at most $2tc$ vertices in $C_i$ form triangles with vertices outside $C_i$.
We call a clique in $C$ a good clique if it has a number of vertices that is not a multiple of 3, otherwise, we call it a bad clique.  
We may assume that there are at most $tc$ good cliques, since each good clique needs at least one vertex in $X$ to form a triangle. Now, we update $X$ by adding all good cliques into it, which is still a twin-cover of $G$. The number of vertices in all good cliques is $O(tc^2)$, so $|X|=O(tc^2)$. 
Thus, we may assume that we are given an instance $(G, k)$ of \textsc{Triangle Partition($tc$)} together with a twin-cover $X$ of $G=(V, E)$, where $|X| = O(tc^2)$, such that the partition $C = \{C_1, \ldots, C_m\}$ of $V \setminus X$ into cliques has only bad cliques.
The next result will be the main tool to relate \textsc{Triangle Partition($tc$)} and \cfa.

\begin{claim}
\label{instance-for-tip}
$G$ has a triangle partition iff there exists a partition $A = \{A_1, A_2, \ldots, A_l\}$ of $X$ into subsets of size $3$ or $6$, a map $f : A \rightarrow C \cup \{\emptyset\}$,  and weight $w : A \rightarrow \mathbb{N}$, such that 
\begin{enumerate}
    \item 
the following holds for each $A_i \in A$:
if $f(A_i) = \emptyset$, then $G[A_i]$ is a triangle;
otherwise, $f(A_i) = C_j$ for some $C_j$ satisfying $A_i \subseteq N(C_j)$.  

Also, if $|A_i| = 3$ and $G[A_i]$ has at least one edge, then $w(A_i) = 3$; 
if $|A_i| = 3$ and $G[A_i]$ has no edge, then $w(A_i) = 6$; 
if $|A_i| = 6$, then $G[A_i]$ has a matching of size $3$ and $w(A_i) = 3$,
\item 
for each $C_j \in C$, we have 
$\sum_{A_i \in f^{-1}(C_j)} w(A_i) \leq |C_j|$.
\end{enumerate}

\end{claim}
\begin{proof}
($\Rightarrow$) Let $T$ be a triangle partition of $G$.  We construct the partition $A$ and maps $f, w$.  For each $T_i \in T$ such that $T_i \subseteq X$, add $T_i$ to $A$ and put $f(T_i) = \emptyset$ ($w(T_i)$ does not matter).  This satisfies the desired conditions since $G[T_i]$ is a triangle.
Let $X'$ be the vertices of $X$ that are in a triangle of $T$ that also includes vertices in a clique of $C$. 
It remains to partition $X'$.  
For every clique $C_j$, consider the set $P_j = \{ C_j \cap T_i : T_i \in T, |C_j \cap T_i| \in \{1, 2\} \} \setminus \{\emptyset\}$.  Thus $P_j$ contains the ``portion of triangles'' that involve $C_j$, and the other portion of each such triangle is in $X'$. 
Note that the sum of sizes of the sets of $P_j$ must be a multiple of $3$, say $3q$ for some integer $q$. In addition, assume $P_j$ has $x$ sets of size 1 and $y$ sets of size 2. Then $3q = x + 2y$. 
(1) Repeatedly do the following process until $\min \{x,y\} = 0$: if there exist two elements $T_i\cap C_j, T_{i'} \cap C_j$ in $P_j$ such that their sizes are 1 and 2, respectively, then add $B = (T_i \cup T_{i'}) \setminus C_j$ to $A$ and both $x,y$ decrease by 1.  Moreover, put $f(B) = C_j$ and $w(B) = 3$; (2) Repeatedly do the following process until $x = 0$ ($y = 0$): if $x \neq 0$ (if $y \neq 0$), then, for any three elements $T_i\cap C_j, T_{i'}\cap C_j, T_{i''}\cap C_j$ in $P_j$, add $B = (T_i \cup T_{i'} \cup T_{i''}) \setminus C_j$ to $A$ and $x$ ($y$) decrease by 3.  Moreover, put $f(B) = C_j$ and $w(B) = 6$ ($w(B) = 3$). 
We can decrease $q$ by 1 during each repeat for the first step and the case $x\neq 0$ of the second step. For the case $y\neq 0$ of the second step, $3q = 2y$ means $q$ is an even number, so we can decrease $q$ by 2 during each repeat. Thus, after the two steps, we have $3q=0$ and all elements of $P_j$ are properly arranged. Moreover, $|B| = 3$ and $G[B]$ has at least one edge in the process (1), $|B| = 3$ and $G[B]$ has no edge in the case $y\neq 0$ of the process (2), $|B| = 6$ and $G[B]$ has a matching of size $3$  in the case $x\neq 0$ of the process (2).
Since $w(B)$ in the processes always equals 3 times the decrease of $q$, $\sum_{A_i \in f^{-1}(C_j)} w(A_i)\leq 3q$, which is at most $|C_j|$.

($\Leftarrow$) Suppose that $A, f$, and $w$ as in the statement exist.  We construct a triangle partition $T$.
For each $A_i$ such that $f(A_i) = \emptyset$, $G[A_i]$ is a triangle and we can add $A_i$ to $T$.
Consider $A_i$ such that $f(A_i) = C_j$ for some $j$.
Suppose that $|A_i| = 3$, with $A_i = \{x, y, z\}$.
If $|A_i| = 3$ and $G[A_i]$ has at least one edge, say $xy$, add to $T_i$ two triangles, one formed with $xy$ and a vertex of $C_j$, the other formed with $z$ and two vertices of $C_j$.  This uses $w(A_i) = 3$ vertices of $C_j$.
If $|A_i| = 3$ and $G[A_i]$ has no edge, add to $T_i$ three triangles, one for each of $x, y, z$, each using two vertices of $C_j$.  This uses $w(A_i) = 6$ vertices of $C_j$.
Suppose that $|A_i| = 6$.  Take a matching of size $3$ in $G[A_i]$ and use the three edges to form three triangles, each edge of the matching using a vertex of $C_j$.  This uses $w(A_i) = 3$ vertices of $C_j$. 
We do this for every $A_i$.  Since $\sum_{A_i \in f^{-1}(C_j)} w(A_i) \leq |C_j|$ for every $C_j$, each clique has enough vertices to use for the above process.
Moreover, every vertex of $X$ is included in a triangle in the above process, since $A_i$ partitions $X$.  Finally, note that the above uses, for each $C_j$, a number of vertices that is a multiple of $3$.  Therefore, the unused vertices of each clique can be split into triangles. 
\end{proof}

Now, for each instance of \textsc{Triangle Partition($tc$)}, 
we describe a corresponding instance of \cfa, which is used to decide whether a partition as described above exists.  We construct a graph $H = (B \cup S, E')$, functions $p, w, c$, conflicting pairs $P$ and integer $q$.
The vertices of $B$ and $S$ are in correspondence with some subsets of $V(G)$.  
First define $B = \binom{X}{3} \cup \binom{X}{6}$.  Add to $P$ the pair $\{b_1, b_2\} \in \binom{B}{2}$ if and only if the sets $b_1$ and $b_2$ have a non-empty intersection.  For each $b \in B$, assign the profit $p(b) = |b|$ and the weight
\begin{align*}
    w(b) = \begin{cases}
        3 & \mbox{if $|b| = 3$ and $G[b]$ has at least one edge;} \\
        6 & \mbox{if $|b| = 3$ and $G[b]$ has no edge;} \\
        3 & \mbox{if $|b| = 6$ and $G[b]$ has a matching of size $3$;} \\
        \infty &\mbox{otherwise.}
    \end{cases}
\end{align*}    
We next define $S$.  First, add each clique $C_i$ of $C$ as a vertex of $S$.  Then, for each triple $t = \{x, y, z\} \in \binom{X}{3}$ such that $G[t]$ is a triangle, add a vertex called $D_t$ to $S$.  
Summarizing, we have $S = C \cup \{D_t:  t \subseteq X$ and $G[t]$ induces a triangle$\}$.
For each $C_i \in C$, put the capacity $c(C_i) = |C_i|$ and for each triangle $t$ of $X$, put the capacity $c(D_t) = 3$. 
Now, define the edges $E'$ of $H$.  
Consider $t \in \binom{X}{3}  \cup \binom{X}{6}$, which is a vertex of $B$.  
If $G[t]$ is a triangle, add an edge from the vertex $t$ in $B$ to the vertex $D_t$ in $S$. Next, assume $G[t]$ is not a triangle.
Let $C(t) \subseteq C$ be the set of cliques of $C$ that contain $t$ in their neighborhood.
Add an edge from $t \in B$ to every $C_i \in C(t)$ that is in $S$.  
Finally, define $q = |X|$.
This completes the construction.

We show that $G$ admits a triangle partition iff the constructed instance of \cfa~admits a solution.
Suppose that $G$ admits a partition $T$ into triangles.  
Let $A$ be a partition of $X$ that satisfies Claim \ref{instance-for-tip}, with map $f$ and weights $w'$.  
Construct a set $F \subseteq E'$ of edges of $H$ as follows.
We define 
\begin{align*}
F = &\{ \{A_i, D_{A_i}\} : (A_i \in A) \wedge  (f(A_i) = \emptyset)\} \cup 
    \{ \{A_i, f(A_i)\} : (A_i \in A) \wedge (f(A_i) \in C)\}.
\end{align*}
Clearly, each $b \in B$ is incident to at most one edge of $F$.
Because $A$ is a partition, no two $A_i$'s intersect, and thus no conflicting pair of vertices of $B$ is incident to an edge of $F$. 
Consider each $s \in S$. If $s$ does not incident with any edge in $F$, then $\sum_{b : bs \in F} w(b) = 0  < c(s)$.
If $s$ is some $D_{A_i}$ such that $\{A_i, D_{A_i}\} \in F$, then $c(s) = 3$ since $G[A_i]$ is a triangle, and  $\sum_{b : bs \in F} w(b) = w(A_i) = 3$ since $G[A_i]$ has at least one edge. Thus, $\sum_{b : bs \in F} w(b) \leq c(s)$.
If $s$ is some $f(A_i)$ such that $\{A_i, f(A_i)\} \in F$, then 
\begin{align*}
\sum_{b : bs \in F} w(b) = \sum_{A_i \in f^{-1}(s)} w(A_i) = \sum_{A_i \in f^{-1}(s)} w'(A_i) \leq |s|  = c(s)
\end{align*}
by Claim \ref{instance-for-tip}. 
Consider the profit of the assigned buyers. Since $q = |X|$, $A$ partitions $X$, and $p(A_i) = |A_i|$ for each $A_i$, we have
\begin{align*}
\sum_{b \in V(F) \cap B} p(b) = \sum_{A_i\in A} p(A_i)  = \sum_{A_i\in A} |A_i| = |X| = q.
\end{align*}
Thus, $F$ is a solution to \cfa.

Conversely, suppose that there is a solution $F \subseteq E'$ to \cfa. Next, we construct a partition $A$  of $X$, a map $f$, and weights function $w'$ that satisfies Claim \ref{instance-for-tip}. 
Let $A = \{A_i \in B : A_i$ is incident to an edge of $F\}$.  By our construction of the conflicting pairs $P$, no two $A_i$'s intersect.  Moreover, since the profit is at least $q = |X|$, the sum of sizes of the $A_i$'s must be $|X|$, and thus $A$ partitions $X$. In addition, the size of each $A_i$ is either 3 or 6 since the size of any element of $B$ is 3 or 6. 
We define weight $w': A \rightarrow \mathbb{N}$ as follows: for any $A_i \in A$, $w'(A_i) = w(A_i)$. 
We define map $f: A \rightarrow C\cup \{\emptyset\}$ as follows: if $A_i \in A$ and $G[A_i]$ is a triangle, then $f(A_i) = \emptyset$; if $A_i \in A$ and $G[A_i]$ is not a triangle, then $f(A_i) = C_j$ such that $C_j \in C$ and $\{A_i, C_j\} \in F$. 
Consider $f$ and $w'$. Since $\emptyset \not\in C$, $G[A_i]$ is a triangle if $f(A_i) = \emptyset$.
Assume $f(A_i) = C_j$ for some $C_j \in C$. Then, $\{A_i, C_j\} \in F \subseteq E$ and $A_i \subseteq N(C_j)$ in $G$. 
$A_i$ must be one of the following three types: (1) $|A_i| = 3$ and $G[A_i]$ has at least one edge, (2) $|A_i| = 3$ and $G[A_i]$ has no edge, (3) $|A_i| = 6$ and $G[A_i]$ has a matching of size $3$, otherwise, $w(A_i) = \infty$ and $w(A_i) > c(C_j)$, a contradiction.
Let $|A_i| = 3$ and $G[A_i]$ has at least one edge. Then $w'(A_i) = w(A_i) = 3$.
Let $|A_i| = 3$ and $G[A_i]$ has no edge. Then $w'(A_i) = w(A_i) = 6$. 
Let $|A_i| = 6$. Then, $A_i$ must be the type (3). Thus, $G[A_i]$ has a matching of size $3$ and $w'(A_i) = w(A_i) = 3$.
According to the definitions of $w'$ and $f$ and the fact that each seller is assigned buyers of at most its
capacity in $H$, we have, for each $C_j \in C$,
\begin{align*}
\sum_{A_i \in f^{-1}(C_j)} w'(A_i) = \sum_{A_i \in f^{-1}(C_j)} w(A_i) = \sum_{\{A_i, C_j\} \in F} w(A_i) \leq c(C_j) = |C_j|.
\end{align*}
Thus, the constructed $A$, $f$, and $w'$ satisfy the requirements of Claim \ref{instance-for-tip}.

Clearly, in the produced instance, the size of $B$ is $tc^{O(1)}$. Based on Lemma \ref{kernel-cfa-para-B-full-app} and the Karp reduction above, we have a compression from \textsc{Triangle Partition($tc$)} to \cfa~of size $|B|^{O(1)} = tc^{O(1)}$. Moreover, since \cfa~is in NP, 
there exists a Karp reduction from \cfa~to the NP-complete problem \textsc{Triangle Partition($tc$)}. Thus, \textsc{Triangle Partition($tc$)} has a PK.
\end{proof}

\begin{theorem}
 \textsc{Clique($tc$)} has a PTK.
\end{theorem}
\begin{proof}
Given an instance $(G, k)$ of \textsc{Clique($tc$)} together with a minimum twin-cover $T$ of $G=(V,E)$. Then $G[V\setminus T]$ is a cluster. 
For every connected component $C$ of $G[V\setminus T]$, produce a new instance $(G[T\cup \{v_C\}],k)$ for \textsc{Clique($tc$)}, where $v_C\in V(C)$. The number of the connected components of $G[V\setminus T]$, which is at most $|V|$, equals the number of the produced new instances each of whose size equals $|T|+1 = tc + 1$. Clearly, $G$ has a clique of size at least $k$ iff at least one of the produced new graphs $G[T\cup \{v\}]$ contains a clique of size at least $k - |V(C)| +1$.
\end{proof}

\begin{theorem}
\textsc{Vertex Cover($tc$)} and \textsc{Odd Cycle Transversal($tc$)} have PKs.
\end{theorem}
\begin{proof}
Assume problem $Q$ is either \textsc{Vertex Cover($tc$)} or \textsc{Odd Cycle Transversal($tc$)}. Given an instance $(G, k)$ of $Q$ together with a minimum twin-cover $T$ of $G=(V,E)$. Assume $C$ is a connected component (a complete graph) of $G[V\setminus T]$.
Let $Q$ be \textsc{Vertex Cover($tc$)}. Since at least $|V(C)| - 1$ vertices of $C$ are in the solution, we repeatedly do the following process for every $C$ in $G[V\setminus T]$ with more than one vertex: delete all but one vertex of $C$ and $k = k - |V(C)| + 1$. Now, parameter $|T|$ is at least the vertex cover number of $G$. Thus, we can obtain a $2 tc$ kernel by using the kernelization for $\textsc{Vertex Cover}(vc)$ \cite{DBLP:journals/jal/ChenKJ01}.
Let $Q$ be \textsc{Odd Cycle Transversal($tc$)}. Suppose $|V(C)|\geq 3$. Since at least $|V(C)| - 2$ vertices of $C$ are in the solution, we repeatedly do the following process for every $C$ in $G[V\setminus T]$ with more than two vertices:  delete all but two vertices of $C$ from $G$, $k = k - |V(C)| + 2$. Now, parameter $|T|$ is at least the feedback vertex set number of $G$. Thus, we can obtain a polynomial kernel by using the kernelization for $\octp$ in parameter feedback vertex set \cite{DBLP:conf/iwpec/JansenK11}.
\end{proof}

\begin{corollary}
\textsc{Independent Set($tc$)} has a PK.
\end{corollary}

\section{Parameterization by neighborhood diversity}
Every graph property expressible in monadic second-order logic (MSO$_1$) has a PC parameterized by $nd$ \cite{DBLP:journals/jcss/GanianSS16}. We provide a meta-theorem to provide PCs for the problems in Table \ref{summary-table-kernels} that are not covered by MSO$_1$.
We say a decision problem $Q$ is a \textit{typical graph problem} if the instance of $Q$ is $(G,k)$, where $G$ is an undirected graph and $k$ is the parameter that can be encoded in $O(\log |G|)$ bits. 
A folklore theorem on parameterized complexity says that every FPT problem has a kernel. Analogously, we can show that every problem that can be solved in $2^{O(nd^c)} |G|^{O(1)}$ time admits PC when parameterized by $nd$.

\begin{theorem}
\label{core PC lemma on nd}	
Let $Q$ be a \textit{typical graph problem} that admits a $2^{O(nd^c)} |G|^{O(1)}$ time algorithm for some constant $c > 0$. Then $Q$ has a compression of bitlength $O(nd^{c+1} + nd^2)$. 
\end{theorem}
\begin{proof}
We provide a compression from $Q$ to a problem $Q'$ which will be defined later. 
Let $(G,k)$ be an instance of $Q$, where $G=(V,E)$.  We can obtain the minimum neighborhood partition $P$ of $V$ in polynomial time. If $nd^c \leq \log |V|$, then decide the input using the $2^{O(nd^c)} |G|^{O(1)} = |G|^{O(1)}$ time algorithm. Assume $nd^c \geq \log |V|$ henceforth. Consider the quotient graph $G_{/P} =(V', E')$, where a vertex $v_M$ of $V'$ is corresponding to the type $M$ in $P$. First, label $v_M$ with $0$ if $M$ is an independent type, and label $v_M$ with $1$ if $M$ is a clique type. Secondly, assign a weight $|M|$ to $v_M$ for each $v_M \in V'$, which can be encoded in $\log |M| \leq nd^c$ bits. Clearly, the labeled quotient graph can be encoded in $O(nd^{c+1} + nd^2)$ bits, because the edges and vertices of $G_{/P}$ can be encoded in $O(nd^2)$ bits, and the labels of the vertices of $G_{/P}$ can be encoded in $O(nd^{c+1})$ bits. In addition, we say $G$ is the original graph of the labeled quotient graph. We now define the problem $Q'$ as follows. The input of the problem is a labeled quotient graph and a parameter $k$, the objective is to decide whether the original graph of the labeled quotient graph together with the parameter $k$ is a yes instance of $Q$. By the definition of typical graph problems, parameter $k$ can be encoded in  $O(\log |G|) = O(nd^c)$ bits. Thus, the size of the instance of $Q'$ is  $O(nd^{c+1} + nd^2)$ bits.
\end{proof}

Since every graph property expressible in MSO$_1$ can be solved in time 
$2^{O(nd)}|G|^{O(1)}$ \cite{DBLP:journals/algorithmica/Lampis12}, the result of Theorem \ref{core PC lemma on nd} for PC covers problems that are potentially outside MSO$_1$.
The following problem can be solved in $2^{O(nd)}|G|^{O(1)}$ time: 
\hcp \cite{DBLP:journals/talg/BergerKMV20, DBLP:journals/jcss/KowalikLNSW22, DBLP:journals/algorithmica/Lampis12}, \cvcp \cite{DBLP:journals/tcs/BergougnouxK19}, \cdsp \cite{DBLP:journals/tcs/BergougnouxK19}, \stp \cite{DBLP:journals/tcs/BergougnouxK19}, \dsp \cite{DBLP:conf/mfcs/BodlaenderLRV10}, \octp \cite{DBLP:journals/corr/abs-2201-04678}, \vcp \cite{DBLP:journals/algorithmica/FominLMT18}, \isp \cite{DBLP:journals/algorithmica/FominLMT18}, \cp \cite{DBLP:journals/algorithmica/FominLMT18}, \fvsp \cite{DBLP:journals/algorithmica/FominLMT18}, \imp \cite{DBLP:journals/algorithmica/FominLMT18}, and \cnp \cite{DBLP:conf/iwpec/GajarskyLO13}, where the facts that $nd(G) \leq cw(G) +1$ \cite{DBLP:journals/algorithmica/Lampis12} and $mw(G)\leq nd(G)$ are needed when some results of the references are used. 
In addition, \pitp \cite{DBLP:journals/dam/Knop20} can be solved in $2^{nd^{O(1)}}|G|^{O(1)}$ time.

\begin{corollary} \label{quadratic compression coro}
Parameterized by neighborhood diversity, the following problems have quadratic compressions: \stp, \imp, \cnp, \hcp, \cdsp, \dsp, \cp, \isp, \vcp, \fvsp, \octp, \cvcp.
Moreover, \pitp  admits a PC.
\end{corollary}

Note that \cp, \vcp, and \isp do not admit compressions of size $O(nd^{2-\epsilon})$ unless NP $\subseteq$ coNP/poly \cite{DBLP:journals/jacm/DellM14}. 
In addition, it is worth mentioning that, in Theorem \ref{core PC lemma on nd}, if $Q$ is NP-complete and $Q'$ is in NP, then $Q$ admits a PK according to Theorem 1.6 of \cite{fomin2019kernelization}. Thus, PKs for some problems here, such as \cp, are obtained straightforwardly. 
We do not specify them here.

\section{Parameterization by modular-width}

We first introduce the concept of \textit{and/or-cross-composition} proposed in \cite{DBLP:journals/siamdm/BodlaenderJK14}, which is a technique for proving polynomial compression lower bounds.

\begin{definition}[Polynomial equivalence relation \cite{DBLP:journals/siamdm/BodlaenderJK14}]
An equivalence relation $R$ on $\Sigma^*$ is called a \textit{polynomial equivalence relation} if the following two conditions hold:
 \begin{enumerate}
 \item
there is an algorithm that given two strings $x,y\in \Sigma^*$ decides whether $x$ and $y$ belong to the same equivalence class in $(|x|+|y|)^{\Oh(1)}$ time;
 \item
for any finite set $S \subseteq \Sigma^*$ the equivalence relation $R$ partitions the elements of $S$ into at most $(\max_{x\in S}|x|)^{\Oh(1)}$ classes.
\end{enumerate}
\end{definition}

\begin{definition}[And-cross-composition (or-cross-composition) \cite{DBLP:journals/siamdm/BodlaenderJK14}]
Let $L\subseteq \Sigma^*$ be a set and let $Q\subseteq \Sigma^* \times \mathbb{N}$ be a parameterized problem. We say that $L$ \textit{and-cross-composes} (\emph{or-cross-composes}) into $Q$ if there is a polynomial equivalence relation $R$ and an algorithm which, given $t$ strings $x_1,\ldots,x_t$ belonging to the same equivalence class of $R$, computes an instance $(y,k)\in \Sigma^* \times \mathbb{N}$ in time polynomial in $\sum_{i=1}^{t} |x_i|$ such that:
 \begin{enumerate}
\item 
the instance $(y,k)$ is yes for $Q$ iff all instances $x_i$ are yes for $L$ (at least one instance $x_i$ is yes for $L$);
\item 
the parameter value $k$ is bounded by a polynomial in $\max_{i=1}^{t} |x_i|+\log t$.
\end{enumerate}
\end{definition}

\begin{theorem}[\cite{DBLP:journals/siamdm/BodlaenderJK14}]
Let $L$ be an NP-hard problem under Karp reductions. If $L$ and/or-cross-composes into a parameterized problem $Q$, then $Q$ does not admit a PC unless NP $\subseteq$ coNP/poly.
\end{theorem}

We provide PC lower bounds for each problem in this section using the and/or-cross-composition technique \cite{DBLP:journals/siamdm/BodlaenderJK14}. For each and/or-cross-composition from $L$ to $Q$ and its related polynomial equivalence relation $R$ on $\Sigma^*$ in this paper, there will be an equivalent class under $R$, which is called a bad class, including all strings each of which is not a valid instance of $L$. Since the bad class can be handled trivially, we may assume that $\Sigma^*$ only includes valid instances of $L$.

\begin{theorem}
\label{clique-no-pc-mw}
\textsc{Clique($mw$)} does not admit a PC unless NP $\subseteq$ coNP/poly.
\end{theorem}
\begin{proof}
Let $\omega (G)$ represent the clique number of $G$. We provide an or-cross-composition from \textsc{Clique} to \textsc{Clique($mw$)}. Assume any two instances $(G_1,k_1)$, $(G_2,k_2)$ of \textsc{Clique} are equivalent under $R$ iff $|V(G_1)|=|V(G_2)|$ and $|k_1|=|k_2|$. Obviously, $R$ is a polynomial equivalence relation. 
Consider the or-cross-composition. Given $t$ instances $(G_1,k),\ldots,(G_t,k)$ of \textsc{Clique} in an equivalence class of $R$, where $|V(G_i)|=n$ for all $i\in [t]$. Produce $(G',k)$ in $\Oh(tn^2)$ time, where $G' = \bigcup_{i=1}^{t}G_i$. Clearly, $mw(G') = \max_{1\leq i\leq t} mw(G_i) \leq n$ and $\omega (G')=\max_{1\leq i\leq t} \omega (G_i)$. Thus, at least one $G_i$ with $\omega (G_i)\geq k$ iff $\omega (G')\geq k$.  
\end{proof}

\begin{corollary}
\textsc{Independent Set($mw$)} and \textsc{Vertex Cover($mw$)} do not admit PCs unless NP $\subseteq$ coNP/poly.
\end{corollary}

\begin{theorem}
\textsc{Chromatic Number($mw$)} does not admit a PC unless NP $\subseteq$ coNP/poly.
\end{theorem}
\begin{proof}
We can provide an and-cross-composition from \cnp to \textsc{Chromatic Number($mw$)}. The reduction goes a similar way to that of Theorem \ref{clique-no-pc-mw}, in which the output instance is the disjoint union of all the input instances, and the only difference is the use of \emph{and}-cross-composition here. 
\end{proof}

\begin{theorem}
\textsc{Hamiltonian Path($mw$)} does not admit a PC unless NP $\subseteq$ coNP/poly.
\end{theorem}
\begin{proof}
We and-cross-composes \textsc{Hamiltonian Path} into \textsc{Hamiltonian Path($mw$)}. Assume any two instances $G_1$, $G_2$ are equivalent under $R$ iff $|V(G_1)|=|V(G_2)|$ and $mw(G_1)=mw(G_2)$. Clearly, $R$ is a polynomial equivalence relation. Assume vertex set $X=\{v_1,$ $\ldots$ $,$  $v_{t-1}\}$.
Consider the and-cross-composition. Given $t$ instances $G_1,\ldots,G_t$ of \textsc{Hamiltonian Path} in an equivalence class of $R$, where $G_i=(V_i,E_i)$, $|V_i|=n$ and $mw(G_i)=mw$ for all $i\in [t]$. Produce an instance $G=(V,E)$ in $\Oh(tn^2)$ time, where  $S=\bigcup_{i=1}^{t}V_i$, $V=S \cup X$, and $E$ is $\bigcup_{i=1}^{t}E_i$ together with all possible edges $uv$ such that $u \in S$ and $v\in X$.
Clearly, the root $v_{V}$ of the modular decomposition tree $MD(G)$ is series since $\overline{G}$ is not connected. Moreover, $v_{V}$ has two children $v_{X}$ and $v_{S}$, each of which is a parallel vertex. Furthermore, all children of $v_{X}$ are leaves of $MD(G)$ and all children of $v_{S}$ are the roots of modular decomposition trees $MD(G_i)$ for all $G_i$. Therefore, $mw(G) = mw \leq n$. Next, we prove that every $G_i$ contains a Hamiltonian path iff $G$ contains a Hamiltonian path.

For the forward direction, suppose every $G_i$ contains a Hamiltonian path $L_i$, and $x_i$, $y_i$ are the ending vertices of $L_i$. For all $1\leq i\leq t-1$, connect $y_i$ of $L_i$ with $v_i$, and connect $v_i$ with $x_{i+1}$ of $L_{i+1}$. Then, we obtain a new path $L = (V, E')$, where $E' = \bigcup_{i=1}^{t-1}(E(L_i) \cup\{y_iv_i, v_ix_{i+1}\}) \cup E(L_t)$. Clearly, $L$ is a Hamiltonian path of $G$.

For the reverse direction, suppose $G$ contains a Hamiltonian path $L=(V,E')$, where $x$, $y$ are the ending vertices of $L$. For a subgraph $S$ of $G$, $L_{S}$ denotes the subgraph induced in $L$ by $V(S)$. 
Assume $L'$ is an induced subgraph of $L$, and $\sigma (L')$ denotes the number of all the edges of $L$ with exactly one endpoint in $L'$. (Note that we consider only the edges in $L$ whenever we use the function $\sigma$.) Clearly, $\sigma (L')=\sigma (L'')$ if $L''$ is the subgraph induced in $L$ by $V\setminus V(L')$.
Let $H$ be a proper induced subgraph of $G$ and $K = G - H$. Since $V(K)\neq \emptyset$ and $L$ is connected, $ \sigma (L_H) \geq 1$. Since the degree of any vertex of $L$ is at most two, we have $ \sigma (L_H)  \leq 2|V(H)|$. More specifically, we consider the following three cases.
First, assume $x,y \in V(H)$. $K$ contains at least an internal vertex of $L$. So $L_H$ contains at least two paths, each of which has at least one ending vertex that is adjacent to a vertex of $K$ to ensure the connectivity of $L$, so $ \sigma (L_H) \geq 2$. The degrees of $x, y$ are one in $L$, so $ \sigma (L_H)  \leq 2|V(H)|-2$.	
Secondly, assume $x,y \in V(K)$. We have $2\leq  \sigma (L_H)  \leq 2|V(H)|$ since $ \sigma (L_H) = \sigma (L_K) \geq 2$. Moreover, $L_H$ is a Hamiltonian path of $H$ if $ \sigma (L_H) =2$.
Thirdly, assume exactly one of $x, y$ is in $V(H)$. Without loss of generality, suppose $x\in V(H)$. We have $1\leq  \sigma (L_H)  \leq 2|V(H)|-1$ since the degree of $x$ in $L$ is one. Moreover, $L_H$ is a Hamiltonian path of $H$ if $ \sigma (L_H) =1$.
Now, consider the subgraphs $G_1,\ldots,G_t$, and $G[X]$ of $G$. According to the construction of $G$, we always have $ \sigma(L_{G[X]}) =\sum_{i=1}^{t} \sigma(L_{G_i}) $.
Assume $x,y \in X$. $ \sigma(L_{G[X]}) $ $\leq$ $2|X|-2$ $\leq$ $2t-4$, but $2t \leq \sum_{i=1}^{t} \sigma(L_{G_i})$ since $ \sigma(L_{G_i}) \geq 2$ for all $i$, a contradiction.
Assume exactly one of $x, y$ is in $X$. Without loss of generality, assume $x \in X$ and $y\in V_1$. Then, we have $ \sigma(L_{G[X]}) $ $\leq$ $2|X|-1$ $\leq 2t-3$, but $ 2t-1 \leq \sum_{i=1}^{t} \sigma(L_{G_i})$ since $ \sigma(L_{G_1}) $ $\geq 1$ and $ \sigma(L_{G_i}) $ $\geq 2$ for all $i\neq 1$, a contradiction.
Suppose $x, y \in V\setminus X$. Then $\sigma(L_{G[X]}) \leq 2|X| $  $\leq $  $ 2t-2$. Assume $x\in V_j$ and $y\in V_k$, where $j,k\in [t]$. If $j=k$, then $x, y \in V_j$ and $\sigma(L_{G_i}) $ $\geq 2$ for all $i$. Thus, we have $ \sigma(L_{G[X]}) $ $< 2t$ $\leq \sum_{i=1}^{t} \sigma(L_{G_i}) $, a contradiction. If $j\neq k$, then $ \sigma(L_{G_j}) $ $\geq 1$, $ \sigma(L_{G_k}) $ $\geq 1$, and $ \sigma(L_{G_i}) $ $\geq 2$ for all $i\neq j,k$. Thus, $ \sigma(L_{G[X]}) $ $\leq 2t-2 \leq$ $\sum_{i=1}^{t} \sigma(L_{G_i}) $. Additionally, since $ \sigma(L_{G[X]}) =\sum_{i=1}^{t} \sigma(L_{G_i}) $, we have $\sum_{i=1}^{t} \sigma(L_{G_i}) $  $= 2t-2$. Consequently, $ \sigma(L_{G_j}) $ $= 1$, $ \sigma(L_{G_k}) $ $= 1$, and $ \sigma(L_{G_i}) $ $= 2$ for all $i\neq j,k$. Consider the subgraph $G_j$. Since $x \in G_j$ and $ \sigma(L_{G_j}) $ $= 1$, $L_{G_j}$ is a Hamiltonian path of $G_j$. Similarly, $L_{G_k}$ is a Hamiltonian path of $G_k$.  Consider $G_i$ for $i\neq j,k$. Since $x,y\not \in G_i$ and $ \sigma(L_{G_j}) $ $= 2$, $L_{G_i}$ is a Hamiltonian path of $G_i$. 
\end{proof}

A PPT reduction from \textsc{Hamiltonian Path($mw$)} to \textsc{Hamiltonian Cycle($mw$)} is a routine: add a vertex $v$ as well as the edges between $v$ and all vertices of the input graph to form the output graph.

\begin{corollary}
\textsc{Hamiltonian Cycle($mw$)} does not admit a PC unless NP $\subseteq$ coNP/poly.
\end{corollary}

Inspired by the use of refinement problems in \cite{Bodlaender_2009}, we define \textsc{Vertex Cover Refinement} as follows: the input is a graph $G$ and a vertex cover set $C$ of $G$, decide whether $G$ has a vertex cover set of size $|C|-1$. 

\begin{lemma}
\label{vcr-nph}
\textsc{Vertex Cover Refinement} is NP-hard under Karp reduction.
\end{lemma}
\begin{proof}
We provide a Karp reduction from \vcp to it. 
Given an instance $(G',k)$ of \vcp, where $G = (V',E')$. Without loss of generality, we may assume $k <|V'|-1$. Construct a new graph $G=(V,E)$, where $V$ is the union of $V'$ and a new set $V''$ with $|V'|-1-k$ vertices, and $E$ consists of all edges of $E'$ and all $uv$ such that $u\in V'$ and $v\in V''$. Clearly, $V'$ is a vertex cover of $G$. Hence, $(G,V')$ is an instance of \textsc{Vertex Cover Refinement}. Assume $(G',k)$ is a yes instance of \vcp. There is a vertex cover $S'$ of  $G'$ such that $|S'|\leq k$. Furthermore, $S'\cup V''$ with at most $k + |V'|- 1 - k = |V'|-1$ vertices is a vertex cover of $G$. For the other direction, assume $(G,V')$ is a yes instance of \textsc{Vertex Cover Refinement}. There is a vertex cover $S$ of $G$ such that $|S|\leq |V'|-1$. Thus, at least one vertex $u$ of $V'$ is not in $S$. For any $v\in V''$,  $uv$ is not covered by $S$ if $v$ is not in $S$. Thus, all vertices of $V''$ are included in $S$, and $G'$ has a vertex cover of size $|S|-|V''|\leq k$.
\end{proof}

\begin{theorem}
\textsc{Connected Vertex Cover($mw$)} does not admit a PC unless NP $\subseteq$ coNP/poly.
\end{theorem}
\begin{proof}
We or-cross-composes \textsc{Vertex Cover Refinement} into \textsc{Connected Vertex Cover($mw$)}. Assume any two instances $(G_1,C_1)$, $(G_2,C_2)$ are equivalent under $R$ iff $|V(G_1)|=|V(G_2)|$, $|C_1|=|C_2|$, and $mw(G_1)=mw(G_2)$. Clearly, $R$ is a polynomial equivalence relation.

Consider the or-cross-composition.  Given $t$ instances $(G_1,C_1),\ldots,(G_t,C_t)$ of \textsc{Vertex Cover Refinement} in an equivalence class of $R$, where $|V(G_i)|=n$, $|C_i|=k$, and $mw(G_i)=mw$ for all $i\in [t]$. Produce an instance $(G',kt)$ in $\Oh(tn^2)$ time, where $G'=(V',E')$, $V' = \bigcup_{i=1}^{t}V(G_i) \cup \{v\}$, and $E'$ equals $\bigcup_{i=1}^{t}E(G_i)$ together with all possible $vw$ for $w\in V'\setminus\{v\}$. Clearly, $mw(G') = mw \leq n$. 
For one direction, assume at least one input instance is yes, say   $(G_t,C_t)$. There exists a vertex set $C$ with size at most $k-1$ that is a vertex cover of $G_t$. Clearly, $C'$ $=$ $C_1\cup\ldots\cup C_{t-1}\cup C\cup \{v\}$ with at most $kt$ vertices is a connected vertex cover of $G'$.
For the other direction, assume $(G',kt)$ is a yes instance. There exists a set $C'$ with at most $kt$ vertices that is a connected vertex cover of $G'$. If $k\geq n$, then any $k-1$ vertices of $G_i$ are a vertex cover of $G_i$ for every $i$. If $k\leq n-1$, then $v$ is in $C'$, otherwise, we need all vertices of $V(G')\setminus \{v\}$, which includes $nt>kt$ vertices, to cover all the edges that are incident with $v$. Thus, $C'\setminus \{v\}$ with at most $kt-1$ vertices comes from the subgraphs $G_1,\ldots,G_t$ of $G'$ and covers the edges inside these subgraphs. As a result, there is at least one $G_i$ containing a vertex cover of size at most $k-1$. 
\end{proof}

We define a new problem named \textsc{Feedback Vertex Set Refinement} as follows: the input is a graph $G$ and a feedback vertex set $F$ of $G$, decide whether $G$ has a feedback vertex set of size $|F|-1$?

\begin{lemma}
\label{fvsr-nph}
\textsc{Feedback Vertex Set Refinement} is NP-hard under Karp reduction.
\end{lemma}
\begin{proof}
The original problem of the Karp reduction is \fvsp. Given an instance $(G,k)$ of \fvsp, where the graph $G=(V,E)$ and the vertex set $V=\{v_1,\ldots,v_n\}$. Without loss of generality, assume $k <n-1$. 
Construct a graph $G'=(V',E')$ as follows. First, add $G$ and vertex sets $U =\{u_1,$ $\ldots,$ $u_{n-1-k}\}$, $W=\{w_{i,j} \mid$ $ 1\leq i\leq n-1-k$ and $1\leq j\leq n\}$ into $G'$. Secondly, for every $u_i$ of $U$, connect $u_i$ with all vertices of $G$. Thirdly, for every edge $u_iv_j$, connect $w_{i,j}$ with the two endpoints of $u_iv_j$. Clearly, $V$ is a feedback vertex set of $G'$. Thus, $(G',V)$ is an instance of \textsc{Feedback Vertex Set Refinement}.
Assume $(G,k)$ is a yes instance. There exists a feedback vertex set $F$ of $G$ such that $|F|\leq k$. Suppose forest $T$ is the subgraph in $G$ induced by $V\setminus F$. The subgraph in $G'$ induced by $V'\setminus (F\cup U)$ is a forest which is generated from $T$ by adhering $n-1-k$ leaf vertices $w_{1,j}, w_{2,j}, \ldots ,w_{n-1-k,j}$ to each vertex $v_j$ of $T$. Therefore, $F\cup U$ with size at most $n-1$  is a feedback vertex set of $G'$.
For the other direction, assume $(G',V)$ is a yes instance. There is a feedback vertex set $F'$ of $G'$ such that $|F'|\leq n-1$. For all $u_i$, they are included in $n$ triangles of $G$, which are $u_i-w_{1,1}-v_1$, $u_i-w_{1,2}-v_2$, $\ldots, u_i-w_{1,n}-v_n$. If $u_i$ is not in $F'$, then $F'$ contains at least one vertex of each triangle other than $u_i$. In addition, apart from $u_i$, the vertices of all the $n$ triangles are different. Thus, $F'$ contains at least $n$ vertices, a contradiction. Thus, $U \subseteq F'$ and $G$ has a feedback vertex set of size $|F'|-|U|\leq k$.
\end{proof}

\begin{theorem}
\label{fvs-no-polykernel}
\textsc{Feedback Vertex Set($mw$)} does not admit a PC unless NP $\subseteq$ coNP/poly.
\end{theorem}
\begin{proof}
We or-cross-compose \textsc{Feedback Vertex Set Refinement} into \textsc{Feedback Vertex Set($mw$)}. Assume any two instances $(G_1,F_1), (G_2,F_2)$ are equivalent under $R$ iff $|V(G_1)|=|V(G_2)|$, $|F_1|=|F_2|$, and $mw(G_1)=mw(G_2)$. Clearly, $R$ is a polynomial equivalence relation. 
Consider the or-cross-composition. Suppose $F_{fvs} (G)$ denotes the feedback vertex number of $G$. Given $t$ instances $(G_1,F_1),\ldots,(G_t,F_t)$ of \textsc{Feedback Vertex Set Refinement} in an equivalence class of $R$, where $|V(G_i)|=n$, $|F_i|=k$, and $mw(G_i)=mw$ for all $i\in [t]$. Produce an instance $(G,kt-1)$ in $\Oh(tn^2)$ time, where $\bigcup_{i=1}^{t}G_i = G = (V,E)$. Clearly, $mw(G)=mw \leq n$ and $F_{fvs} (G)=\sum_{i=1}^{t} F_{fvs} (G_i)$. Thus, $F_{fvs} (G_i)\leq k-1$ for at least one $G_i$ iff $F_{fvs} (G) \leq kt-1$.
\end{proof}

We define a new problem named \textsc{Odd Cycle Transversal Refinement} as follows:
the input is a graph $G$ and an odd cycle transversal $O$ of $G$, decide whether $G$ has an odd cycle transversal of size $|O|-1$?

\begin{theorem}
\textsc{Odd Cycle Transversal($mw$)} has no PCs unless NP $\subseteq$ coNP/poly.
\end{theorem}
\begin{proof}
We first show that \textsc{Odd Cycle Transversal Refinement} is NP-hard under Karp reduction. The process is in the same way as that of Theorem \ref{fvsr-nph}. Then, we demonstrate that \textsc{Odd Cycle Transversal($mw$)} does not admit a PC unless NP $\subseteq$ coNP/poly. The process is in the same way as that of Theorem \ref{fvs-no-polykernel}.
\end{proof}

\textsc{Induced Matching} is NP-complete \cite{DBLP:journals/siamdm/KoS03}. The size of an induced matching is the number of edges of the induced matching. We define a new problem named \textsc{Induced Matching Refinement} as follows: the input is a graph $G$ and an induced matching of $G$ whose size is $k$, decide whether $G$ has an induced matching of size $k+1$?

\begin{lemma}
\label{IM-nph}
\textsc{Induced Matching Refinement} is NP-hard under Karp reduction.
\end{lemma}
\begin{proof}
The original problem of the Karp reduction is \imp. Given an instance $(G,k)$ of \textsc{Induced Matching}, where $G=(V,E)$ and $V=\{v_1,\ldots,v_n\}$. Without loss of generality, we may assume $2\leq k \leq 0.5n$. Construct a graph $G'=(V',E')$ as follows. First, add $G$ and all vertices of the sets $U =\{u_1,\ldots,u_n\}$, $W =\{w_1,\ldots,w_n\}$, $X =\{x_1,\ldots,x_{n-k+1}\}$ into $G'$. Then, for every $u_i$ of $U$, connect $u_i$ with all vertices of $V\cup \{w_i\}$. Finally, connect $x_i$ with $w_i$ for every $i\in [n-k]$ and connect $x_{n-k+1}$ with all vertices of $\{w_{n-k+1},\ldots,w_n\}$. Clearly, edge set  $Y=\{u_1w_1,\ldots,u_nw_n\}$ is an induced matching of $G'$. Thus, $(G',Y)$ is an instance of \textsc{Induced Matching Refinement}of size $n$.

Suppose $(G,k)$ is a yes instance. There exists $I\subseteq V$ with $2k$ vertices such that the subgraph induced by $I$ is a matching of size $k$. Consider graph $G'$. Suppose $W' =\{w_1,\ldots,w_{n-k+1}\}$. 
The subgraph induced in $G'$ by $I\cup W' \cup X$ is an matching with size $k+(n-k+1)=n+1$. Thus, $(G',Y)$ is a yes instance of \textsc{Induced Matching Refinement}.
For the other direction, assume $(G',Y)$ is a yes instance. There exists $I'\subseteq V'$ with $2(n+1)$ vertices such that the subgraph induced by $I'$ is a matching of size $n+1$. We first use proof by contradiction to show that $I'\cap U$ is an empty set as follows. Suppose $I'\cap U$ contains a vertex $u_i$. Then $N(u_i) \cap I'$ are either $\{w_i\}$ or $\{v\}$, where $v\in V$.
Assume $N(u_i) \cap I'=\{v\}$. Then $I' \subseteq L=V'\setminus N(\{u_i,v\})$ and the maximum induced matching (MIM) of $G'[L]$ equals that of $G'$. Obviously, if $1\leq i \leq n-k$, then the size of the MIM of $G'[L]$ is $n-k+1 \leq n-1$. If $n-k+1\leq i \leq n$, then the size of the MIM of $G'[L]$ is $n-k+2 \leq n$. This is a contradiction.  
Assume $N(u_i) \cap I'= \{w_i\}$. Then $I'\subseteq L=V'\setminus N(\{u_i,w_i\})$ and the MIM of $G'[L]$ equals that of $G'$. Clearly, the size of the MIM of $G'[L]$ is $n$, a contradiction.
Now, we know  $I' \subseteq V\cup W\cup X$. Clearly, the subgraph induced in $G'$ by $V\cup W\cup X$ consists of $n-k$ independent edges, the subgraph $G$, and a star with $k$ degrees. Since $G'$ has an induced matching of size $n+1$ and the size of the MIM of a star is at most one, there is an induced matching of $G$ whose size is at least $(n+1)-(n-k)-1=k$. Hence, $(G,k)$ is a yes instance of \imp.
\end{proof}

\begin{theorem}
\textsc{Induced Matching($mw$)} does not admit a PC unless NP $\subseteq$ coNP/poly.
\end{theorem}
\begin{proof}
We can provide an or-cross-composition from \textsc{Induced Matching Refinement} to \textsc{Induced Matching($mw$)}. The reduction goes the same way as that of Theorem \ref{fvs-no-polykernel}, in which the output instance is the disjoint union of all the input instances.
\end{proof}

We define \textsc{Dominating Set Refinement} problem, which is NP-complete \cite{Bodlaender_2009}, as follows: the input is a graph $G$ and a dominating set $D$ of $G$, decide whether $G$ has a dominating set of size $|D|-1$?

\begin{theorem}
\label{ds-no-pk}
\textsc{Dominating Set($mw$)} does not admit a PC unless NP $\subseteq$ coNP/poly.
\end{theorem}
\begin{proof}
We can provide an or-cross-composition from \textsc{Dominating Set Refinement} to \textsc{Dominating Set($mw$)}. The reduction goes the same way as that of Theorem \ref{fvs-no-polykernel}, in which the output instance is the disjoint union of all the input instances.
\end{proof}

\section{Conclusions}
We conclude the paper by proposing an open question. Fomin et. al. state in the open problems chapter of their kernelization textbook \cite{fomin2019kernelization}: ``Finding an example demonstrating that polynomial
compression is a strictly more general concept than polynomial kernelization, is an extremely interesting open problem.'' Inspired by this, we propose the following question: do there exist quadratic kernels for the problems in Corollary \ref{quadratic compression coro} parameterized by neighborhood diversity? Even quadratic Turing kernels will be interesting.

\section*{Acknowledgement}

We thank the anonymous reviewers for their valuable comments.

\bibliographystyle{abbrv}
\bibliography{references1}
\end{document}